\documentclass[%
 reprint, amsmath,amssymb,superscriptaddress,
 aps]{revtex4-1}
\usepackage[latin1]{inputenc}
\usepackage{mathrsfs}
\usepackage{enumerate}
\usepackage{amsmath}
\usepackage{booktabs}
\usepackage{array}
\usepackage{amsthm}
\usepackage{float}
\usepackage{dcolumn}
\usepackage[colorlinks=true, linkcolor=blue, citecolor=blue, urlcolor=blue ]{hyperref}
\usepackage{bm}
\usepackage{longtable}
\usepackage{graphicx}
\usepackage{epsfig}
\usepackage{epstopdf}
\usepackage{amssymb}
\usepackage{color}

\begin{document}
\title{Characterizing entanglement shareability and distribution in $N$-partite systems}
\author{Hui Li}
\affiliation{School of Mathematical Sciences, Hebei Normal University, Shijiazhuang 050024, China}
\affiliation{School of Science, Hebei University of Science and Technology, Shijiazhuang 050018, China}
\author{Ting Gao}
\email{gaoting@hebtu.edu.cn}
\affiliation{School of Mathematical Sciences, Hebei Normal University, Shijiazhuang 050024, China}
\author{Fengli Yan}
\email{flyan@hebtu.edu.cn}
\affiliation{College of Physics, Hebei Key Laboratory of Photophysics Research and Application, Hebei Normal University, Shijiazhuang 050024, China}
\begin{abstract}
Exploring the shareability and distribution of entanglement possesses fundamental significance in quantum information tasks. In this paper, we demonstrate that the square of bipartite entanglement measures $G_q$-concurrence, which is the generalization of concurrence, follows a set of hierarchical monogamy relations for any $N$-qubit quantum state. On the basis of these monogamy inequalities, we render two kinds of hierarchical indicators that exhibit evident advantages in the capacity of witnessing entanglement. Moreover, we show an analytical relation between $G_q$-concurrence and concurrence in $2\otimes d$ systems. Furthermore, we rigorously prove that the monogamy property of squared $G_q$-concurrence is superior to that of squared concurrence in $2\otimes d_2\otimes d_3\otimes\cdots\otimes d_N$ systems. In addition, several concrete examples are provided to illustrate that for multilevel systems, the squared $G_q$-concurrence satisfies the monogamy relation, even if the squared concurrence does not. These results better reveal the intriguing characteristic of multilevel
entanglement and provide critical insights into the entanglement distribution within multipartite quantum systems.
\end{abstract}



\maketitle

\section{Introduction}\label{I}
Multipartite quantum entanglement is an important resource that has demonstrated pronounced advantages in quantum information tasks such as quantum computation \cite{PRL87.047901,PRSLSA459.2011,PRL110.060504}, quantum simulation \cite{PMP86.153}, and quantum sensing \cite{RMP89.035002}. The monogamy of entanglement (MoE) is a highly essential characteristic of multipartite quantum systems, which indicates that quantum entanglement cannot be freely shared among all parties \cite{IJRD48.71, RMP81.865}. In addition, the MoE has critical applications in ensuring the security of quantum communication \cite{PRL69.2881,PRL70.1895,EPJB41.75,EPL84.50001,PRA74.022317,PRA93.042324,PRA83.022319}, quantum steering \cite{PRA88.062338,NJP16.083017}, and condensed matter physics \cite{NP7.399}.

Tracing the development of quantitative monogamy relation, starting from the pioneering work of Coffman $et~al$., a rigorous mathematical formulation for MoE based on squared concurrence (SC) was established in three-qubit systems \cite{PRA61.052306}. Subsequently, Osborne and Verstraete extended this relation to the $N$-qubit scenarios \cite{PRL96.220503},
\begin{equation}\label{I0}
\begin{aligned}
C^2(\rho_{A_1|A_2\cdots A_N})\geq C^2(\rho_{A_1A_2})+\cdots+C^2(\rho_{A_1A_N}),
\end{aligned}
\end{equation}
where
$C(\rho_{A_1|A_2\cdots A_N})$ is the bipartite entanglement of quantum state $\rho$ under the splitting $A_1|A_2\cdots A_N$, and $C(\rho_{A_1A_i})$ is the concurrence of $\rho_{A_1A_i}$ with $i=2,3,\cdots,N$. Later studies have employed diverse entanglement measures, such as squashed entanglement \cite{PRA69.022309,JMP45.829}, negativity \cite{PRA79.012329}, entanglement of formation \cite{PRL113.100503,PRA90.062343}, Tsallis-$q$ entanglement \cite{PRA81.062328,PRA93.062340}, and R\'{e}nyi-$\alpha$ entanglement \cite{PRA93.022306}, to describe the entanglement distribution of multipartite quantum systems \cite{PRA73.022325,PRA77.032329,PRL103.050501,PRA87.032330,PRL114.140402,PRA92.062345,PRL117.060501,PRA101.032301,QIP21.357,PRA110.012405,PRA86.062323,PRL112.180501,PRA109.012213}.

Furthermore, we proposed a generalized class of entanglement measures called $G_q$-concurrence, and proved that the squared $G_q$-concurrence (S$G_q$C) satisfies the monogamy relation
\begin{equation}\label{I1}
\begin{aligned}
\mathscr{C}_q^2(\rho_{A_1|A_2\cdots A_N})\geq\mathscr{C}_q^2(\rho_{A_1A_2})+\cdots+\mathscr{C}_q^2(\rho_{A_1A_N})
\end{aligned}
\end{equation}
for $1<q\leq2$ \cite{arXiv:2406}. The merit of this result lies in the fact that the entanglement indicators produced by S$G_q$C ($1<q<2$) can detect all genuinely multiqubit entangled states, thereby compensating the shortcomings of the entanglement indicator generated by SC. We also demonstrated that $G_q$-concurrence and concurrence can be linked through an analytical expression in $2\otimes2$ systems \cite{arXiv:2406}. To characterize the interesting entanglement features of multipartite quantum states more deeply, it is highly desirable to further examine whether such an analytical relation still exists in high-dimensional systems.

It is acknowledged that the SC obeys a family of hierarchical $k$-partite monogamy inequalities in $N$-qubit systems \cite{PRL96.220503}
\begin{equation}\label{I2}
\begin{aligned}
C^2(\rho_{A_1|A_2\cdots A_N})\geq \sum\limits_{i=2}^{k-1}C^2(\rho_{A_1A_i})+C^2(\rho_{A_1|A_k\cdots A_N}),
\end{aligned}
\end{equation}
which serves as a tool to probe entanglement in $k$-partite case with $3\leq k\leq N$. This naturally prompts us to consider the following questions: (Q1) Does S$G_q$C satisfy an analogous monogamy relation? (Q2) Are the monogamy properties of S$G_q$C and SC in the multipartite systems equivalent? (Q3) Are the monogamy inequalities of S$G_q$C applicable to multilevel quantum systems? and (Q4) Compared to SC, what advantages does the monogamy property of S$G_q$C have?

The structure of this paper is as follows. In Sec.~\ref{II}, we introduce some fundamental knowledge. In Sec.~\ref{III}, a set of hierarchical monogamy relations is provided for $N$-qubit quantum systems divided into $k$ parties. When $k=N$, the hierarchical $k$-partite monogamy relation is reduced to the result given in Ref. \cite{arXiv:2406}. As an ancillary result, the analytical relation between $G_q$-concurrence and concurrence in $2\otimes d$ systems is established. In Sec.~\ref{IV}, we construct two classes of entanglement indicators. In particular, we elucidate that the hierarchical indicator generated from S$G_q$C can effectively detect the entanglement of $W$ state, whereas its counterpart based on SC lacks this capability. In Sec.~\ref{V}, rigorous proof shows that the S$G_q$C exhibits better monogamy property than SC in $2\otimes d_2\otimes d_3\otimes\cdots\otimes d_N$ systems. To illustrate the applicability of the entanglement indicators of S$G_q$C in multilevel systems, we provide several concrete examples from two different scenarios: one where the first subsystem is qubit, and the other where it is multilevel. In Sec.~\ref{VI}, we make a summary.

\section{Preliminaries}\label{II}
In this section, we introduce the definitions of concurrence and $G_q$-concurrence, as well as the relation between them.

For any bipartite pure state $|\phi\rangle_{AB}$, the concurrence is defined as \cite{RMP81.865,PRL78.5022,PRL80.2245}
\begin{equation}
\begin{aligned}
C(|\phi\rangle_{AB})=\sqrt{2[1-{\rm Tr}(\rho_A^2)]},
\end{aligned}
\end{equation}
the $G_q$-concurrence ($q>1$) is \cite{arXiv:2406}
\begin{equation}
\begin{aligned}
\mathscr{C}_q(|\phi\rangle_{AB})=[1-{\rm Tr}(\rho_{A}^q)]^{\frac{1}{q}}.
\end{aligned}
\end{equation}
Here $\rho_{A}$ is the reduced density operator of $|\phi\rangle_{AB}$ with respect to subsystem $A$.

For any bipartite mixed state $\rho_{AB}$, the concurrence is given by
\begin{equation}
\begin{aligned}
C(\rho_{AB})=\min\limits_{\{p_i,|\phi_{i}\rangle_{AB}\}}\sum\limits_{i}p_iC(|\phi_i\rangle_{AB}),
\end{aligned}
\end{equation}
the $G_q$-concurrence ($q>1$) is
\begin{equation}
\begin{aligned}
\mathscr{C}_{q}(\rho_{AB})=\min\limits_{\{p_i,|\phi_{i}\rangle_{AB}\}}\sum\limits_{i}p_i\mathscr{C}_{q}(|\phi_i\rangle_{AB}).
\end{aligned}
\end{equation}
Here the minimum runs over all possible pure state decompositions $\{p_i,|\phi_i\rangle_{AB}\}$ of $\rho_{AB}$.

In addition, the analytical relation between $G_q$-concurrence and concurrence is \cite{arXiv:2406}
\begin{equation}\label{P0}
\begin{aligned}
\mathscr{C}_q(\rho_{AB})=h_q[C^2(\rho_{AB})]\\
\end{aligned}
\end{equation}
for any $2\otimes2$ quantum state $\rho_{AB}$, where
\begin{equation}\label{P1}
\begin{aligned}
h_q(t)=\Big[1-\Big(\frac{1+\sqrt{1-t}}{2}\Big)^q-\Big(\frac{1-\sqrt{1-t}}{2}\Big)^q\Big]^{\frac{1}{q}},\\
\end{aligned}
\end{equation}
and $1<q\leq2$.

\theoremstyle{remark}
\newtheorem{definition}{\indent Definition}
\newtheorem{lemma}{\indent Lemma}
\newtheorem{theorem}{\indent Theorem}
\newtheorem{proposition}{\indent Proposition}
\newtheorem{corollary}{\indent Corollary}
\newtheorem{example}{\indent Example}
\def\QEDclosed{\mbox{\rule[0pt]{1.3ex}{1.3ex}}}
\def\QED{\QEDclosed}
\def\proof{\noindent{\indent\em Proof}.}
\def\endproof{\hspace*{\fill}~\QED\par\endtrivlist\unskip}

\section{The monogamy relations for S$G_q$C}\label{III}
In this section, we first present the analytical formula between $G_q$-concurrence and concurrence in $2\otimes d$ quantum systems, which lays the foundation for showing that S$G_q$C satisfies the monogamy relations in tripartite systems and a set of hierarchical $k$-partite monogamy relations in $N$-qubit systems.

Before discussing the core results, we first present two lemmas.

\begin{lemma}\label{lemma 1}
The function $h_q(C^2)$ is a concave function with respect to $C^2$ for $1<q\leq2$.
\end{lemma}

\begin{lemma}\label{lemma 2}
The function $h_q^2(C^2)$ is a monotonically increasing and convex function with respect to $C^2$ for $1<q\leq2$.
\end{lemma}

The detailed proofs of Lemmas \ref{lemma 1} and \ref{lemma 2} can be referred to in Appendixes \ref{A} and \ref{B}, respectively.

It is well-established that the $G_q$-concurrence admits an analytical expression for any $2\otimes d$ pure state and $1<q\leq2$ \cite{arXiv:2406}. In the following, we derive an analogous analytical representation for arbitrary $2\otimes d$ mixed states based on Lemma \ref{lemma 1}.

\begin{theorem}\label{theorem 1}
For any $2\otimes d$ mixed state $\rho_{AB}$, the relation between $G_q$-concurrence and concurrence is
\begin{equation}\label{th10}
\begin{aligned}
\mathscr{C}_q(\rho_{AB})=h_q[C^2(\rho_{AB})],
\end{aligned}
\end{equation}
where $1<q\leq2$.
\end{theorem}

\begin{proof}
For an arbitrary $2\otimes d$ mixed state $\rho_{AB}$, suppose that $\{p_i,|\varphi_i\rangle_{AB}\}$ is the optimal pure state decomposition of $\mathscr{C}_q(\rho_{AB})$ such that $\mathscr{C}_q(\rho_{AB})=\sum_ip_i\mathscr{C}_q(|\varphi_i\rangle_{AB})$, and $\{q_j,|\phi_j\rangle_{AB}\}$ is the optimal ensemble decomposition of $C^2(\rho_{AB})$ such that $C^2(\rho_{AB})=\sum_jq_j C^2(|\phi_j\rangle_{AB})$, then one obtains
\begin{equation}\label{th11}
\begin{aligned}
\mathscr{C}_q(\rho_{AB})=&\sum_ip_i\mathscr{C}_q(|\varphi_i\rangle_{AB})\\
=&\sum_ip_ih_q[C^2(|\varphi_i\rangle_{AB})]\\
\leq&\sum_jq_jh_q[C^2(|\phi_j\rangle_{AB})]\\
\leq&h_q\bigg[\sum_jq_j C^2(|\phi_j\rangle_{AB})\bigg]\\
=&h_q[C^2(\rho_{AB})],\\
\end{aligned}
\end{equation}
where the second equality is because $\mathscr{C}_q(|\varphi_i\rangle_{AB})=h_q[C^2(|\varphi_i\rangle_{AB})]$ holds for any $2\otimes d$ pure state \cite{arXiv:2406}, the first inequality is according to the definition of $\mathscr{C}_q(\rho_{AB})$, and the second inequality is true owing to the fact that $h_q(C^2)$ is a concave function of $C^2$ for $1<q\leq2$.

On the other hand, let $\{s_k,|\psi_k\rangle_{AB}\}$ be the optimal decomposition of $C(\rho_{AB})$, and under the optimal pure decomposition  $\{p_i,|\varphi_i\rangle_{AB}\}$  of $\mathscr{C}_q(\rho_{AB})$, there is
\begin{equation}\label{th12}
\begin{aligned}
\mathscr{C}_q(\rho_{AB})=&\sum_ip_i\mathscr{C}_q(|\varphi_i\rangle_{AB})\\
=&\sum_ip_ih_q[C^2(|\varphi_i\rangle_{AB})]\\
\geq&h_q\Big[\Big(\sum_ip_iC(|\varphi_i\rangle_{AB})\Big)^2\Big]\\
\geq&h_q\Big[\Big(\sum_ks_kC(|\psi_k\rangle_{AB})\Big)^2\Big]\\
=&h_q[C^2(\rho_{AB})],\\
\end{aligned}
\end{equation}
where the first inequality holds because $h_q(C^2)$ is a convex function of $C$ for $1<q\leq2$ \cite{arXiv:2406}, the second inequality can be derived since $h_q(C^2)$ is a monotonically increasing function with respect to $C^2$.

Combing Eqs. (\ref{th11}) and (\ref{th12}), we have that the formula (\ref{th10}) is valid.
\end{proof}

Next, we present a kind of monogamy relations in tripartite quantum systems.

\begin{theorem}\label{theorem 2}
For any $2\otimes 2\otimes 2^{N-2}$ quantum state $\rho_{AB\mathbb{C}}$, the square of $G_q$-concurrence possesses the following relation
\begin{equation}\label{th30}
\begin{aligned}
\mathscr{C}_q^2(\rho_{A|B\mathbb{C}})\geq\mathscr{C}_q^2(\rho_{AB})+\mathscr{C}_q^2(\rho_{A\mathbb{C}}),\\
\end{aligned}
\end{equation}
where $1<q\leq2$.
\end{theorem}

\begin{proof}
For any tripartite quantum state $\rho_{AB\mathbb{C}}$ in $2\otimes 2\otimes 2^{N-2}$ systems, we have
\begin{equation}\label{th41}
\begin{aligned}
&\mathscr{C}_q^2(\rho_{A|B\mathbb{C}})-\mathscr{C}_q^2(\rho_{AB})-\mathscr{C}_q^2(\rho_{A\mathbb{C}})\\
=&h_q^2[C^2(\rho_{A|B\mathbb{C}})]-h_q^2[C^2(\rho_{AB})]-h_q^2[C^2(\rho_{A\mathbb{C}})]\\
\geq& h_q^2[C^2(\rho_{AB})+C^2(\rho_{A\mathbb{C}})]-h_q^2[C^2(\rho_{AB})]
-h_q^2[C^2(\rho_{A\mathbb{C}})]\\
\geq&0.\\
\end{aligned}
\end{equation}
Here the first equality is obtained by utilizing Eq. (\ref{th10}), the first inequality holds because $h_q^2(C^2)$ is a monotonically increasing function of $C^2$ for $1<q\leq2$ and the SC satisfies the monogamy relation $C^2(\rho_{A|B\mathbb{C}})\geq C^2(\rho_{AB})+C^2(\rho_{A\mathbb{C}})$ \cite{PRL96.220503}, and the second inequality is true based on Lemma \ref{lemma 2}.
\end{proof}

If we divide system $\mathbb{C}$ into single qubit subsystem $C_1$ and $2^{\otimes{N-3}}$ subsystem $\mathbb{C}_2$, then by using the inequality (\ref{th41})  to the quantum state $\rho_{AC_1\mathbb{C}_2}$, one has
\begin{equation*}
\begin{aligned}
\mathscr{C}_q^2(\rho_{A|B\mathbb{C}})-\mathscr{C}_q^2(\rho_{AB})-\mathscr{C}_q^2(\rho_{AC_1})-\mathscr{C}_q^2(\rho_{A\mathbb{C}_2})\geq0.\\
\end{aligned}
\end{equation*}
Similarly, by further successive segmenting the last quantum system and then iteratively using inequality (\ref{th41}), we can obtain a family of hierarchical $k$-partite monogamy relations in $N$-qubit systems, as shown in the following theorem.

\begin{theorem}\label{theorem 3}
For any $N$-qubit quantum state $\rho_{A_1A_{2}\cdots A_N}$, a set of hierarchical $k$-partite monogamy relations is
\begin{equation}\label{th40}
\begin{aligned}
\mathscr{C}_q^2(\rho_{A_1|A_{2}\cdots A_N})\geq\sum\limits_{i=2}^{k-1}\mathscr{C}_q^2(\rho_{A_1A_{i}})+\mathscr{C}_q^2(\rho_{A_1|A_k\cdots A_N}),\\
\end{aligned}
\end{equation}
where $1<q\leq2$.
\end{theorem}

Note that when $q=2$, the inequality (\ref{th40}) can be simplified to inequality (\ref{I2}); when $k=N$, the formula (\ref{th40}) is reduced to formula (\ref{I1}); when $q=2$ and $k=N$, the formula (\ref{th40}) is further reduced to inequality (\ref{I0}).

As a corollary of Theorem \ref{theorem 3}, we establish a family of analogical hierarchical monogamy inequalities in terms of the $\alpha$-th power of $G_q$-concurrence.

\begin{corollary}
For any $N$-qubit quantum state $\rho_{A_1A_{2}\cdots A_N}$, the $\alpha$-th ($\alpha\geq2$) power of $G_q$-concurrence satisfies a class of hierarchical $k$-partite monogamy relations
\begin{equation}\label{th5}
\begin{aligned}
\mathscr{C}_q^{\alpha}(\rho_{A_1|A_{2}\cdots A_N})\geq\sum\limits_{i=2}^{k-1}\mathscr{C}_q^{\alpha}(\rho_{A_1A_{i}})+\mathscr{C}_q^{\alpha}(\rho_{A_1|A_k\cdots A_N}),\\
\end{aligned}
\end{equation}
where $1<q\leq2$.
\end{corollary}

\begin{proof}
We first analyze a quantum state $\rho_{A_1A_2\mathbb{C}}$ in $2\otimes2\otimes2^{N-2}$ systems. From Theorem \ref{theorem 2}, we have $\mathscr{C}_q^2(\rho_{A_1|A_2\mathbb{C}})\geq\mathscr{C}_q^2(\rho_{A_1A_2})+\mathscr{C}_q^2(\rho_{A_1\mathbb{C}})$. Without loss of generality, we suppose that $\mathscr{C}_q^2(\rho_{A_1A_2})\geq\mathscr{C}_q^2(\rho_{A_1\mathbb{C}})$, then there is
\begin{equation}\label{th51}
\begin{aligned}
\mathscr{C}_q^{\alpha}(\rho_{A_1|A_2\mathbb{C}})\geq&[\mathscr{C}_q^{2}(\rho_{A_1A_2})+\mathscr{C}_q^{2}(\rho_{A_1\mathbb{C}})]^{\frac{\alpha}{2}}\\
=&\mathscr{C}_q^{\alpha}(\rho_{A_1A_2})\Big[1+\frac{\mathscr{C}_q^{2}(\rho_{A_1\mathbb{C}})}{\mathscr{C}_q^{2}(\rho_{A_1A_2})}\Big]^{\frac{\alpha}{2}}\\
\geq&\mathscr{C}_q^{\alpha}(\rho_{A_1A_2})\Big[1+\frac{\mathscr{C}_q^{\alpha}(\rho_{A_1\mathbb{C}})}{\mathscr{C}_q^{\alpha}(\rho_{A_1A_2})}\Big]\\
=&\mathscr{C}_q^{\alpha}(\rho_{A_1A_2})+\mathscr{C}_q^{\alpha}(\rho_{A_1\mathbb{C}}).\\
\end{aligned}
\end{equation}
Here the second inequality is obtained according to the relation $(1+x)^\mu\geq1+x^\mu$, $0\leq x\leq1$, $\mu\geq1$.

Then, the last subsystem $\mathbb{C}=A_3A_4\cdots A_n$ is sequentially split and the formula (\ref{th51}) is iteratively applied, which leads to the validity of inequality (\ref{th5}).
\end{proof}

\section{Multipartite entanglement indicators based on S$G_q$C}\label{IV}
For any $N$-partite pure state $|\phi\rangle_{N}$, based on the monogamy relation (\ref{th40}), we construct a class of entanglement hierarchical indicators,
\begin{equation}
\begin{aligned}
\tau_{qk}(|\phi\rangle_{N})=&\mathscr{C}_q^2(|\phi\rangle_{A_1|A_{2}\cdots A_N})-\sum_{i=2}^{k-1}\mathscr{C}_q^2(\rho_{A_1A_{i}})\\
&-\mathscr{C}_q^2(\rho_{A_1|A_k\cdots A_N}).\\
\end{aligned}
\end{equation}
Furthermore, for any $N$-partite mixed state $\rho_N$, two kinds of entanglement hierarchical indicators are established,
\begin{equation}
\begin{aligned}
&\tau_{qk}^{1}{(\rho_{N})}=\min\limits_{\{p_i,|\phi_i\rangle_N\}}\sum_ip_i\tau_{qk}(|\phi_i\rangle_{N}),\\
&\tau_{qk}^{2}{(\rho_{N})}=\mathscr{C}_q^2(\rho_{A_1|A_{2}\cdots A_N})-\sum_{i=2}^{k-1}\mathscr{C}_q^2(\rho_{A_1A_{i}})\\
&~~~~~~~~~~~~~-\mathscr{C}_q^2(\rho_{A_1|A_k\cdots A_N}).\\
\end{aligned}
\end{equation}
Here $1<q<2$, $k=3,4,\cdots,N$, $\tau_{qk}^{1}{(\rho_{N})}$ is defined by convex-roof extension, and $\tau_{qk}^{2}{(\rho_{N})}$ is constructed by means of the monogamy inequality of mixed states. When $k=N$, $\tau_{qk}^{1}{(\rho_{N})}$ is accorded with the entanglement indicator given in Ref. \cite{arXiv:2406}.


Next, we render an example to illustrate the application of entanglement indicator.

\begin{example}
Consider an $N$-qubit $W$ state $|W_N\rangle=\frac{1}{\sqrt{N}}(|10\cdots0\rangle+|01\cdots0\rangle+\cdots+|00\cdots1\rangle)$, then one has that the hierarchical entanglement indicator based on SC is $\tau(|W_{N}\rangle)=C^2(|W\rangle_{A_1|A_2\cdots A_N})-(k-2)C^2(\rho_{A_1A_{2}})-C^2(\rho_{A_1|A_{k}\cdots A_N})$=0, while the entanglement indicator $\tau_{qk}(|W_{N}\rangle)=\mathscr{C}_q^2(|W\rangle_{A_1|A_2\cdots A_N})-(k-2)\mathscr{C}_q^2(\rho_{A_1A_{2}})-\mathscr{C}_q^2(\rho_{A_1|A_{k}\cdots A_N})$, where
$\mathscr{C}_q^2(|W\rangle_{A_1|A_2\cdots A_N})=\Big[1-\Big(\frac{1+\sqrt{1-\frac{4(N-1)}{N^2}}}{2}\Big)^q
-\Big(\frac{1-\sqrt{1-\frac{4(N-1)}{N^2}}}{2}\Big)^q\Big]^\frac{2}{q}, \mathscr{C}_q^2(\rho_{A_1A_{2}})=\Big[1-\Big(\frac{1+\sqrt{1-\frac{4}{N^2}}}{2}\Big)^q
-\Big(\frac{1-\sqrt{1-\frac{4}{N^2}}}{2}\Big)^q\Big]^\frac{2}{q}, \mathscr{C}_q^2(\rho_{A_1|A_{k}\cdots A_N})=\Big[1-\Big(\frac{1+\sqrt{1-\frac{4(N-k+1)}{N^2}}}{2}\Big)^q
-\Big(\frac{1-\sqrt{1-\frac{4(N-k+1)}{N^2}}}{2}\Big)^q\Big]^\frac{2}{q}$. Let $N=8$, the parameter $k$ ranges from $3$ to $8$, and $q$ is taken as $1.3$, $1.4$, $1.5$, $1.6$, and $1.7$, respectively. The Table \ref{t1} shows $\tau_{qk}(|W_{N}\rangle)>0$, which indicates $\tau_{qk}$ can effectively detect the entanglement of $W$ state.
\end{example}

This result states that the monogamy property of S$G_q$C is not equivalent to that of SC when $1<q<2$.
\section{The comparison of the monogamy property between S$G_q$C and SC in multilevel systems}\label{V}
From the proceeding discussion, we can see that both S$G_q$C and SC satisfy the monogamy relation in the $N$-qubit system. This naturally raises the question of whether the monogamy property of S$G_q$C is applicable to multilevel systems, and whether S$G_q$C possesses superior monogamy property compared to SC?

We first discuss the scenario of $2\otimes d_2\otimes d_3\otimes\cdots\otimes d_N$ systems, wherein the first subsystem is two-level and the remaining subsystems are multilevel. The following result can be obtained.

\begin{theorem}
For any quantum state in $2\otimes d_2\otimes d_3\otimes\cdots\otimes d_N$ systems, the monogamy property of squared $G_q$-concurrence exceeds that of squared concurrence.
\end{theorem}
\begin{table}[htbp]
	\newcommand{\tabincell}[2]{\begin{tabular}{@{}#1@{}}#2\end{tabular}}
	\centering
	\caption{\label{t1} For the 8-qubit $W$ state, the values of entanglement indicator $\tau_{qk}(|W_{8}\rangle)$ are presented for $k=3,4,5,6,7,8$ and different parameters $q$.}
	\begin{tabular}{cccccc}
		\toprule
		\hline
		\hline
		\specialrule{0em}{1pt}{1pt}
        ~~&~~~~$q$=1.3~~  & $~~~q$=1.4~~ & $~~~q$=1.5~~  & $~~~q$=1.6~~ & $~~q$=1.7\\
		
		\specialrule{0em}{1pt}{1pt}
		\midrule
		\hline
		\specialrule{0em}{1.5pt}{1.5pt}
		\tabincell{c}{~~$k=3$}&\tabincell{c}{~~0.0031}&\tabincell{c}{~~0.0048}&\tabincell{c}{~~0.0063}&\tabincell{c}{~~0.0070}&\tabincell{c}{~~0.0069}\\
		\specialrule{0em}{1.5pt}{1.5pt}
		\tabincell{c}{~~$k=4$}&\tabincell{c}{~~0.0076}&\tabincell{c}{~~0.0128}&\tabincell{c}{~~0.0181}&\tabincell{c}{~~0.0230}&\tabincell{c}{~~0.0270} \\
		\specialrule{0em}{1.5pt}{1.5pt}
		\tabincell{c}{~~$k=5$}&\tabincell{c}{~~0.0119}&\tabincell{c}{~~0.0203}&\tabincell{c}{~~0.0295}&\tabincell{c}{~~0.0385}&\tabincell{c}{~~0.0466} \\
		\specialrule{0em}{1.5pt}{1.5pt}
		\tabincell{c}{~~$k=6$}&\tabincell{c}{~~0.0158}&\tabincell{c}{~~0.0274}&\tabincell{c}{~~0.0402}&\tabincell{c}{~~0.0532}&\tabincell{c}{~~0.0656} \\
		\specialrule{0em}{1.5pt}{1.5pt}
		\tabincell{c}{~~$k=7$}&\tabincell{c}{~~0.0193}&\tabincell{c}{~~0.0337}&\tabincell{c}{~~0.0501}&\tabincell{c}{~~0.0671}&\tabincell{c}{~~0.0837}\\
		\specialrule{0em}{1.5pt}{1.5pt}
		\tabincell{c}{~~$k=8$}&\tabincell{c}{~~0.0222}&\tabincell{c}{~~0.0391}&\tabincell{c}{~~0.0587}&\tabincell{c}{~~0.0796}&\tabincell{c}{~~0.1005} \\
		\bottomrule
		\specialrule{0em}{1.5pt}{1.5pt}
		\hline
		\hline
	\end{tabular}
\end{table}

\begin{proof}
Given an arbitrary pure state $|\phi\rangle_{A_1A_2\cdots A_N}$ in $2\otimes d_2\otimes d_3\otimes\cdots\otimes d_N$ systems, the entanglement distribution has the following relation
\begin{equation*}
\begin{aligned}
&\mathscr{C}_q^2(|\phi\rangle_{A_1|A_{2}\cdots A_N})-\sum\limits_{i=2}^{N}\mathscr{C}_q^2(\rho_{A_1A_{i}})\\
=&h_q^2[C^2(|\phi\rangle_{A_1|A_{2}\cdots A_N})]-\sum\limits_{i=2}^{N}h_q^2[C^2(\rho_{A_1A_{i}})]\\
=&k_{q1}C^2(|\phi\rangle_{A_1|A_{2}\cdots A_N})-\sum\limits_{i=2}^{N}k_{qi}C^2(\rho_{A_1A_{i}})\\
=&k_{q1}[C^2(|\phi\rangle_{A_1|A_{2}\cdots A_N})-\sum\limits_{i=2}^{N}C^2(\rho_{A_1A_{i}})]+\Delta_1,\\
\end{aligned}
\end{equation*}
where $k_{q1}=\frac{h_q^2[C^2(|\phi\rangle_{A_1|A_{2}\cdots A_N})]}{C^2(|\phi\rangle_{A_1|A_{2}\cdots A_N})}$, $k_{qi}=\frac{h_q^2[C^2(\rho_{A_1A_{i}})]}{C^2(\rho_{A_1A_{i}})}$, $i=2,3,\cdots,N$, $\Delta_1=\sum_{i=2}^N(k_{q1}-k_{qi})C^2(\rho_{A_1A_{i}})$.

If the SC obeys the monogamy relation for any $2\otimes d_2\otimes d_3\otimes\cdots\otimes d_N$ pure state, then there is
\begin{equation*}
\begin{aligned}
\Delta_2=k_{q1}[C^2(|\phi\rangle_{A_1|A_{2}\cdots A_N})-\sum\limits_{i=2}^{N}C^2(\rho_{A_1A_{i}})]\geq0.\\
\end{aligned}
\end{equation*}
In addition, the relation $k_{q1}\geq k_{qi}$  holds because $h_q^2(C^2)$ is a monotonically increasing and convex function with respect to $C^2$, which leads to $\Delta_1\geq0$. Therefore, we obtain
\begin{equation*}
\begin{aligned}
\mathscr{C}_q^2(|\phi\rangle_{A_1|A_{2}\cdots A_N})-\sum\limits_{i=2}^{N}\mathscr{C}_q^2(\rho_{A_1A_{i}})
=\Delta_1+\Delta_2\geq0.\\
\end{aligned}
\end{equation*}

For any mixed state $\rho_{A_1A_{2}\cdots A_N}$, we suppose that $\{p_i,|\psi_i\rangle_{A_1|A_{2}\cdots A_N}\}$ is the optimal pure state decomposition of $\mathscr{C}_q(\rho_{A_1|A_{2}\cdots A_N})$, then there is
\begin{equation*}
\begin{aligned}
\mathscr{C}_q(\rho_{A_1|A_{2}\cdots A_N})&=\sum_ip_i\mathscr{C}_q(|\psi_i\rangle_{A_1|A_{2}\cdots A_N})=\sum_iC_{1_i},\\
\mathscr{C}_q^\prime(\rho_{A_1A_{j}})&=\sum_ip_i\mathscr{C}_q(\rho_{A_1A_j}^i)=\sum_iC_{j_i},\\
\end{aligned}
\end{equation*}
where $\mathscr{C}_q^\prime(\rho_{A_1A_{j}})$ is the average $G_q$-concurrence under some decomposition, $j=2,3,\cdots,N$. Then, we can obtain
\begin{equation}\label{th24}
\begin{aligned}
&\mathscr{C}_q^2(\rho_{A_1|A_{2}\cdots A_N})-\sum_{j=2}^{N}\mathscr{C}_q^2(\rho_{A_1A_{j}})\\
\geq&\Big(\sum_iC_{1_i}\Big)^2-\sum_{j=2}^{N}\Big(\sum_iC_{j_i}\Big)^2\\
=&\sum_i(C_{1_i}^2-\sum_jC_{j_i}^2)+\Delta
\end{aligned}
\end{equation}
where the inequality holds since $\mathscr{C}_q(\rho_{A_1A_{j}})\leq\sum_iC_{j_i}$, the first term in the equality is nonnegative according to the monogamy property of pure state, $\Delta=2\sum_i\sum_{k=i+1}(C_{1_i}C_{1_k}-\sum_{j=2}^{N}C_{j_i}C_{j_k})$. The following we derive $\Delta$ is also greater than or equal to zero. For any pure state $|\psi_i\rangle$ and $|\psi_k\rangle$ chosen from the optimal pure decomposition of $\mathscr{C}_q(\rho_{A_1|A_{2}\cdots A_N})$, one has
\begin{widetext}
\begin{equation*}
\begin{aligned}
C_{1_i}^{2}C_{1_k}^{2}\geq&(C_{2_i}^{2}+C_{3_i}^{2}+\cdots+C_{N_i}^{2})(C_{2_k}^{2}+C_{3_k}^{2}+\cdots+C_{N_k}^{2})\\
=&(C_{2_i}C_{2_k})^2+(C_{2_i}C_{3_k})^2+\cdots+(C_{2_i}C_{N_k})^2
+(C_{3_i}C_{2_k})^2+(C_{3_i}C_{3_k})^2+\cdots+(C_{3_i}C_{N_k})^2\\
&+\cdots+(C_{N_i}C_{2_k})^2+(C_{N_i}C_{3_k})^2+\cdots+(C_{N_i}C_{N_k})^2\\
\geq&(C_{2_i}C_{2_k})^2+(C_{3_i}C_{3_k})^2+\cdots+(C_{N_i}C_{N_k})^2
+2(C_{2_i}C_{2_k})(C_{3_i}C_{3_k})+\cdots+2(C_{2_i}C_{2_k})(C_{N_i}C_{N_k})\\
&+2(C_{3_i}C_{3_k})(C_{4_i}C_{4_k})+\cdots+2(C_{3_i}C_{3_k})(C_{N_i}C_{N_k})
+\cdots+2(C_{{(N-1)}_i}C_{{(N-1)}_k})(C_{N_i}C_{N_k})\\
=&\Big(\sum_{j=2}^{N}C_{j_i}C_{j_k}\Big)^2,
\end{aligned}
\end{equation*}
\end{widetext}
where the first inequality is valid owing to the monogamy property of pure state, and the second inequality utilizes the triangle inequality. Therefore, there is
\begin{equation*}\label{th25}
\begin{aligned}
C_{1_i}C_{1_k}\geq C_{2_i}C_{2_k}+C_{3_i}C_{3_k}+\cdots+C_{N_i}C_{N_k},
\end{aligned}
\end{equation*}
which indicates that $\Delta\geq0$ holds.

Therefore, we obtain that the S$G_q$C must obeys monogamy relation if this property is satisfied by SC.

Next, we demonstrate that S$G_q$C may still meet the monogamy relation even if SC exhibits polygamy $C^2(\rho_{A_1|A_2\cdots A_N})-\sum_jC^2(\rho_{A_1A_j})\leq0$. In this situation, for any $2\otimes d_2\otimes d_3\otimes\cdots\otimes d_N$ pure state, there is
\begin{equation*}
\begin{aligned}
\Delta_2=k_{q1}[C^2(|\phi\rangle_{A_1|A_{2}\cdots A_N})-\sum\limits_{i=2}^{N}C^2(\rho_{A_1A_{i}})]\leq0,\\
\end{aligned}
\end{equation*}
then one has
\begin{equation*}
\begin{aligned}
\mathscr{C}_q^2(|\phi\rangle_{A_1|A_{2}\cdots A_N})-\sum\limits_{i=2}^{N}\mathscr{C}_q^2(\rho_{A_1A_{i}})
=\Delta_1-|\Delta_2|,
\end{aligned}
\end{equation*}
which obeys monogamy relation when $\Delta_1\geq|\Delta_2|$. When S$G_q$C conforms to the monogamy relation for pure states, then it also satisfies the monogamy relation for mixed states, and the proof for this is similar to that of formula (\ref{th24}).

Based on the above analysis, we conclude that the monogamy relation of SC is a sufficient condition for that of S$G_q$C. Remarkably, S$G_q$C can also satisfy the monogamy relation under certain condition, even in the absence of this property for SC. Consequently, the monogamy property of S$G_q$C is superior to SC.
\end{proof}

To illustrate this result more clearly, an example is rendered as follows.

\begin{example}
Considering a four-partite mixed state $\rho_{A_1A_2A_3A_4}$ in $2\otimes d_2\otimes d_3\otimes d_4$ systems, we assume that the squared bipartite concurrences are $C^2(\rho_{A_1|A_2A_3A_4})=0.76$, $C^2(\rho_{A_1A_2})=C^2(\rho_{A_1A_3})=C^2(\rho_{A_1A_4})=0.27$, then one sees
$C^2(\rho_{A_1|A_2A_3A_4})-C^2(\rho_{A_1A_2})-C^2(\rho_{A_1A_3})-C^2(\rho_{A_1A_4})=-0.05$, whereas the entanglement indicator generated by S$G_q$C is valid for this state since $\tau_{q4}^2(\rho_{A_1|A_2A_3A_4})=\mathscr{C}_{q}^2(\rho_{A_1|A_2A_3A_4})-\mathscr{C}_{q}^2(\rho_{A_1A_2})-\mathscr{C}_{q}^2(\rho_{A_1A_3})-\mathscr{C}_{q}^2(\rho_{A_1A_4})
=0.0229$ when the parameter $q$ is chosen as $\frac{3}{2}$. In this case, S$G_q$C satisfies the monogamy relation, while SC does not.
\end{example}

\begin{figure}[htbp]
\centering
{\includegraphics[width=8cm,height=6cm]{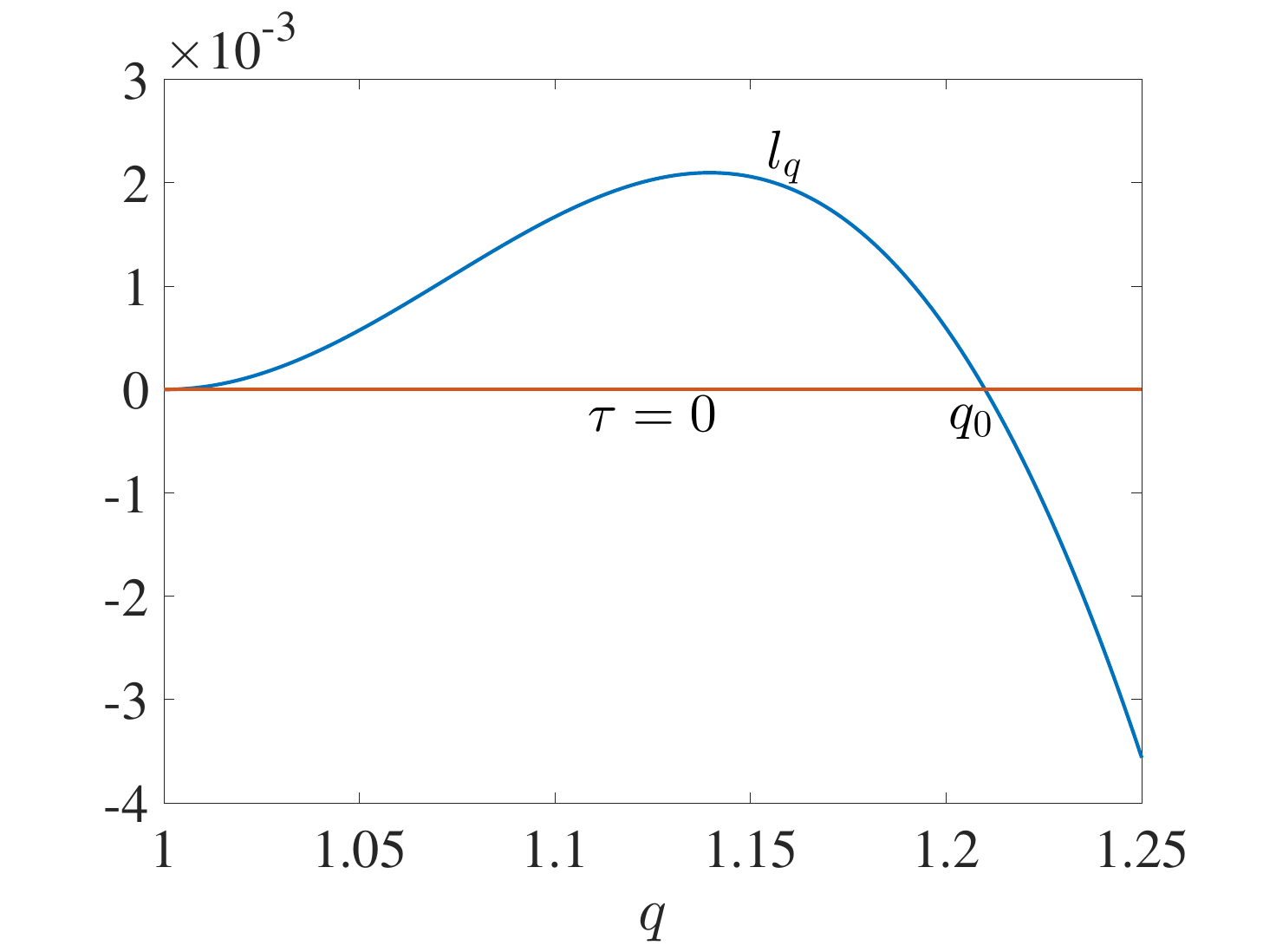}}
\caption{The blue line $l_q$ is a lower bound of $\tau_{q3}(|\psi\rangle_{A|BC})$, the red line denotes $\tau(|\psi\rangle_{A|BC})=0$, the intersection point of these two lines is denoted as $q_0$.\\}\label{fig 5}
\end{figure}
We now turn to the monogamy properties of S$G_q$C and SC within a multipartite quantum system where the first subsystem is multilevel.

\begin{example}
Consider a $3\otimes3\otimes3$ quantum state $|\psi\rangle_{ABC}=\frac{1}{\sqrt{6}}(|123\rangle-|132\rangle+|231\rangle-|213\rangle+|312\rangle-|321\rangle)$, Ou pointed out that based on SC, the state $|\psi\rangle_{ABC}$ does not satisfy the monogamy property since $C^2(|\psi\rangle_{A|BC})-C^2(\rho_{AB})-C^2(\rho_{AC})=-\frac{2}{3}<0$ \cite{PRA75.034305}. However, based on S$G_q$C, we have $\tau_{q3}(|\psi\rangle_{A|BC})=\mathscr{C}_{q}^2(|\psi\rangle_{A|BC})-\mathscr{C}_{q}^2(\rho_{AB})-\mathscr{C}_{q}^2(\rho_{AC})\geq[1-(\frac{1}{3})^{q-1}]^{\frac{2}{q}}-2\times[1-2\times(\frac{1}{2})^q]^{\frac{2}{q}}=l_q$.
From Fig. \ref{fig 5}, we can see that S$G_q$C must obey monogamy relation for $q\in(1,q_0)$.
\end{example}

\begin{example}
Consider a $4\otimes2\otimes2$ pure state $|\varphi\rangle_{ABC}=\frac{1}{\sqrt{2}}(a|000\rangle+b|110\rangle+a|201\rangle+b|311\rangle)$ with $a=\sin\theta, b=\cos\theta$ \cite{PRA90.062343}. The reduced density operator for subsystem $AB$ is $\rho_{AB}=\frac{1}{2}|\phi_1\rangle\langle\phi_1|+\frac{1}{2}|\phi_2\rangle\langle\phi_2|$, where $|\phi_1\rangle=a|00\rangle+b|11\rangle$ and $|\phi_2\rangle=a|20\rangle+b|31\rangle$. For any pure state decomposition of $\rho_{AB}$, the form of each pure state is
\begin{equation}
\begin{aligned}
|\widetilde{\phi}_i\rangle=a_i|\phi_1\rangle+{\rm e}^{-{\rm i}\gamma}\sqrt{1-a_i^2}|\phi_2\rangle,
\end{aligned}
\end{equation}
for which the reduced density operator with respect to subsystem $B$ is given by $\rho_B^i={\rm diag}\{a^2,b^2\}$. Thereupon, we obtain $C^2(\rho_{AB})=4a^2b^2$ and $\mathscr{C}_q^2(\rho_{AB})=(1-a^{2q}-b^{2q})^{\frac{2}{q}}$. Analogously, for the reduced density operator $\rho_{AC}$, one has $C^2(\rho_{AC})=1$ and $\mathscr{C}_q^2(\rho_{AC})=[1-(\frac{1}{2})^{q-1}]^{\frac{2}{q}}$.
Additionally, the SC and S$G_q$C of $|\varphi\rangle_{ABC}$ are respectively
\begin{figure}[htbp]
\centering
{\includegraphics[width=8cm,height=6cm]{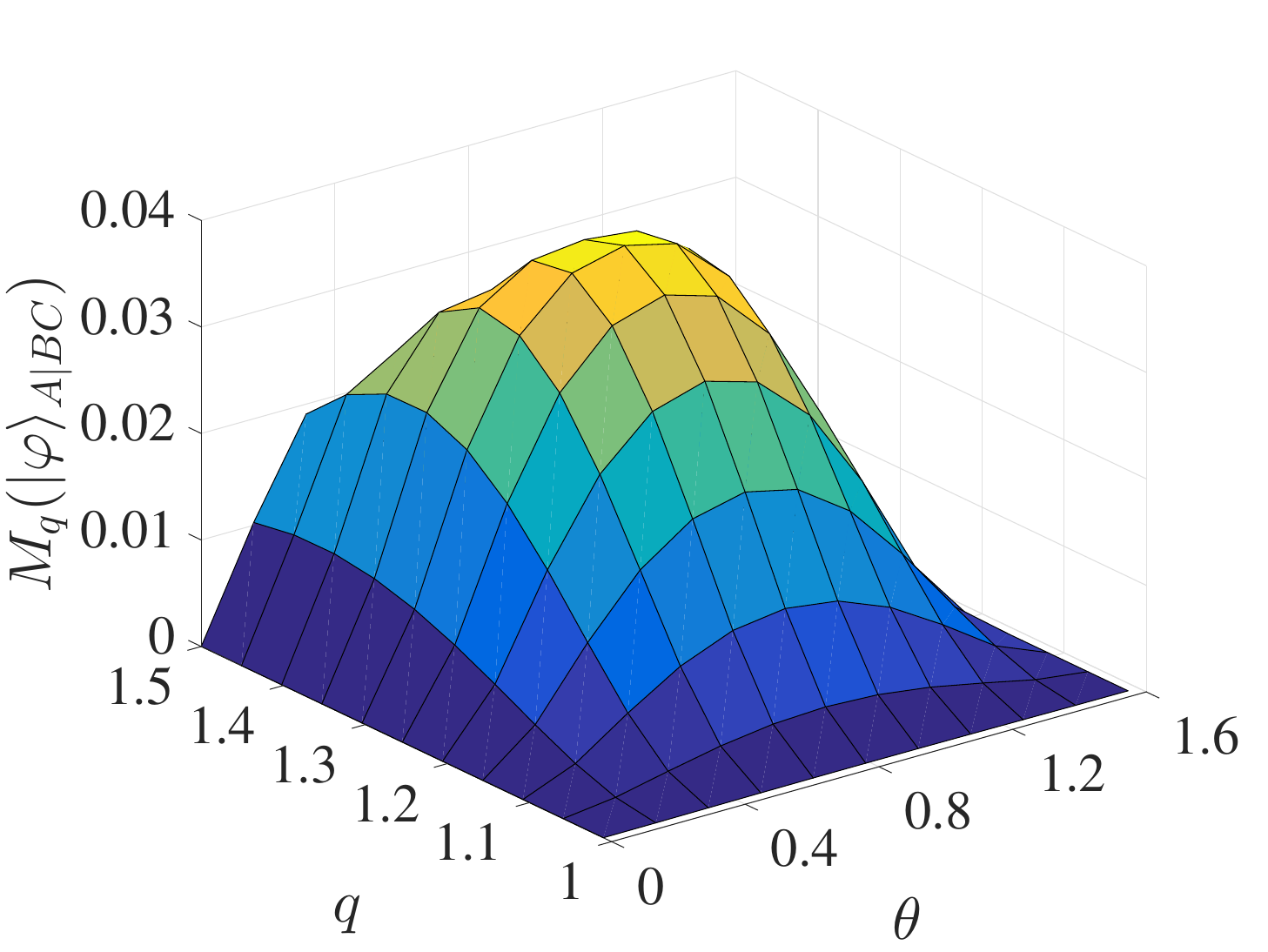}}
\caption{The entanglement distribution $M_q(|\varphi\rangle_{A|BC})=\mathscr{C}_q^2(|\varphi\rangle_{A|BC})-\mathscr{C}_q^2(\rho_{AB})-\mathscr{C}_q^2(\rho_{AC})$ as a function of $\theta~(0\leq\theta\leq\frac{\pi}{2})$ and $q~(1\leq q\leq1.5)$ is nonnegative, which indicates S$G_q$C is monogamous. \\}\label{fig e4}
\end{figure}
\begin{equation}
\begin{aligned}
&C^2(|\varphi\rangle_{A|BC})=2-a^4-b^4,\\
&\mathscr{C}_q^2(|\varphi\rangle_{A|BC})=\Big(1-\frac{a^{2q}}{2^{q-1}}-\frac{b^{2q}}{2^{q-1}}\Big)^{\frac{2}{q}},\\
\end{aligned}
\end{equation}
under the bipartition $A|BC$. Then, one has
\begin{equation*}
\begin{aligned}
M(|\varphi\rangle_{A|BC})&=C^2(|\varphi\rangle_{A|BC})-C^2(\rho_{AB})-C^2(\rho_{AC})\\
&=2-a^4-b^4-4a^2b^2-1\\
&=-2a^2b^2\\
&\leq0.\\
\end{aligned}
\end{equation*}
Obviously, SC does not satisfy the monogamy relation for this state. Meanwhile, we have
\begin{equation*}
\begin{aligned}
M_q(|\varphi\rangle_{A|BC})=&\mathscr{C}_q^2(|\varphi\rangle_{A|BC})-\mathscr{C}_q^2(\rho_{AB})-\mathscr{C}_q^2(\rho_{AC})\\
=&\Big[1-\frac{a^{2q}}{2^{q-1}}-\frac{b^{2q}}{2^{q-1}}\Big]^{\frac{2}{q}}-(1-a^{2q}-b^{2q})^{\frac{2}{q}}\\
&-\big[1-\big(\frac{1}{2}\big)^{q-1}\big]^{\frac{2}{q}},\\
\end{aligned}
\end{equation*}
where the parameters $a,b$ are taken as $a=\sin\theta$ and $b=\cos\theta$, respectively. So we can see that the S$G_q$C obeys monogamy property when $\theta\in[0,\frac{\pi}{2}]$ and $q\in[1,1.5]$, as shown in Fig. \ref{fig e4}.
\end{example}

Therefore, these results indicate that the entanglement indicator produced by S$G_q$C is not only applicable to multilevel quantum systems but also exhibits superior performance in detecting entanglement compared to the entanglement indicator generated by SC. Furthermore, our results can more accurately reveal the entanglement structure of a multipartite quantum state that goes beyond bipartite entanglement and better characterize the entanglement distribution of multipartite quantum states.

\section{Conclusion}\label{VI}
In this work, we have proved that S$G_q$C satisfies a set of hierarchical monogamy relations when an $N$-qubit quantum system is divided into $k$ parties, which is an essential extension of previous results. By utilizing these monogamy inequalities, we have constructed two classes of hierarchical entanglement indicators. Compared with the entanglement indicators produced by squared concurrence, they demonstrate notable advantages in the capability of detecting entanglement. Moreover, in $2\otimes d$ systems, we have established the connection between $G_q$-concurrence and concurrence through an analytical function. In addition, rigorous proof has shown that the squared $G_q$-concurrence possesses better monogamy property than squared concurrence in $2\otimes d_2\otimes d_3\otimes\cdots\otimes d_N$ systems. Several concrete examples, where the first subsystem is elaborated both in a qubit scenario and in a multilevel scenario, illustrate that the S$G_q$C fulfills monogamy relation in multilevel systems even if SC fails. These results contribute to characterizing the entanglement distribution of multipartite quantum states.

\section*{ACKNOWLEDGMENTS}
This work was supported by the National Natural Science Foundation of China under Grant No. 62271189, and the Hebei Central Guidance on Local Science and Technology Development Foundation of China under Grant No. 236Z7604G.

\begin{appendix}
\section{Proof of Lemma 1}\label{A}
Prove that the function $h_q(C^2)$ is concave with respect to $C^2$, that is, prove that its second derivative is non-positive for $1<q\leq2$. Set $t=C^2$, then we have
\begin{equation*}
\begin{aligned}
\frac{\partial h_q}{\partial t}=&\frac{1}{2^{q+1}}\bigg[1-\bigg(\frac{1+\sqrt{1-t}}{2}\bigg)^q-\bigg(\frac{1-\sqrt{1-t}}{2}\bigg)^q\bigg]^{\frac{1}{q}-1}\\
&\times\frac{(1+\sqrt{1-t})^{q-1}-(1-\sqrt{1-t})^{q-1}}{\sqrt{1-t}}.\\
\end{aligned}
\end{equation*}
Define
\begin{equation*}
\begin{aligned}
g_q=\frac{\partial^2h_q}{\partial t^2},
\end{aligned}
\end{equation*}
then there is
\begin{equation}\label{le0}
\begin{aligned}
g_q=&\frac{1}{2^{q+1}}\bigg[1-\bigg(\frac{1+\sqrt{1-t}}{2}\bigg)^q-\bigg(\frac{1-\sqrt{1-t}}{2}\bigg)^q\bigg]^{\frac{1}{q}-2}\\
&\times M(t,q),\\
\end{aligned}
\end{equation}
where $M(t,q)=\frac{1-q}{2^{q+1}}\xi_1+\xi_2(\xi_3-\xi_4)$,
\begin{equation}\label{le1}
\begin{aligned}
&\xi_1=\frac{[(1+\sqrt{1-t})^{q-1}-(1-\sqrt{1-t})^{q-1}]^2}{{1-t}},\\
&\xi_2=1-\bigg(\frac{1+\sqrt{1-t}}{2}\bigg)^q-\bigg(\frac{1-\sqrt{1-t}}{2}\bigg)^q,\\
&\xi_3=\frac{(1+\sqrt{1-t})^{q-2}}{2(1-t)}\bigg[\frac{1+\sqrt{1-t}}{\sqrt{1-t}}-(q-1)\bigg],\\
&\xi_4=\frac{(1-\sqrt{1-t})^{q-2}}{2(1-t)}\bigg[\frac{1-\sqrt{1-t}}{\sqrt{1-t}}+(q-1)\bigg].\\
\end{aligned}
\end{equation}
It is noted that the term before $M(t,q)$ in Eq. (\ref{le0}) is positive for $0<t<1$ and $1<q\leq2$, so judging the sign of $g_q$ is equivalent to judging the sign of $M(t,q)$. Thereof, we will analyze the maximal and minimal values of $M(t,q)$ in the region $D=\{(t,q)|0\leq t\leq1,1<q\leq2\}$, which must be attained at critical points (within the interior of $D$) or boundary points of $D$.
However, there is no solution of
\begin{figure}[htbp]
\centering
{\includegraphics[width=8cm,height=6cm]{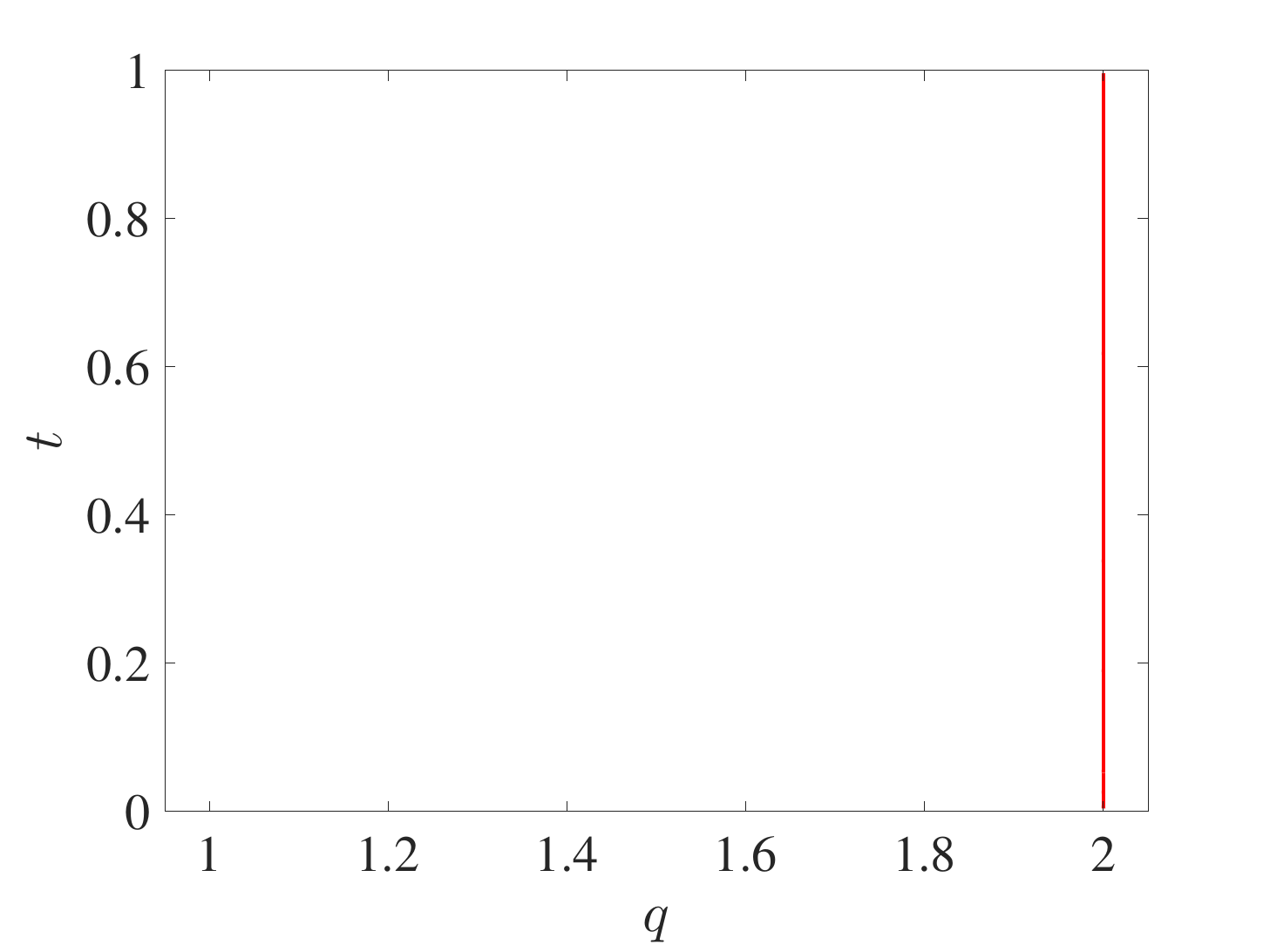}}
\caption{The red line denotes the solutions of $\frac{\partial M(t,q)}{\partial t}=0$.\\}\label{fig 1}
\end{figure}
\begin{figure}[htbp]
\centering
{\includegraphics[width=8cm,height=6cm]{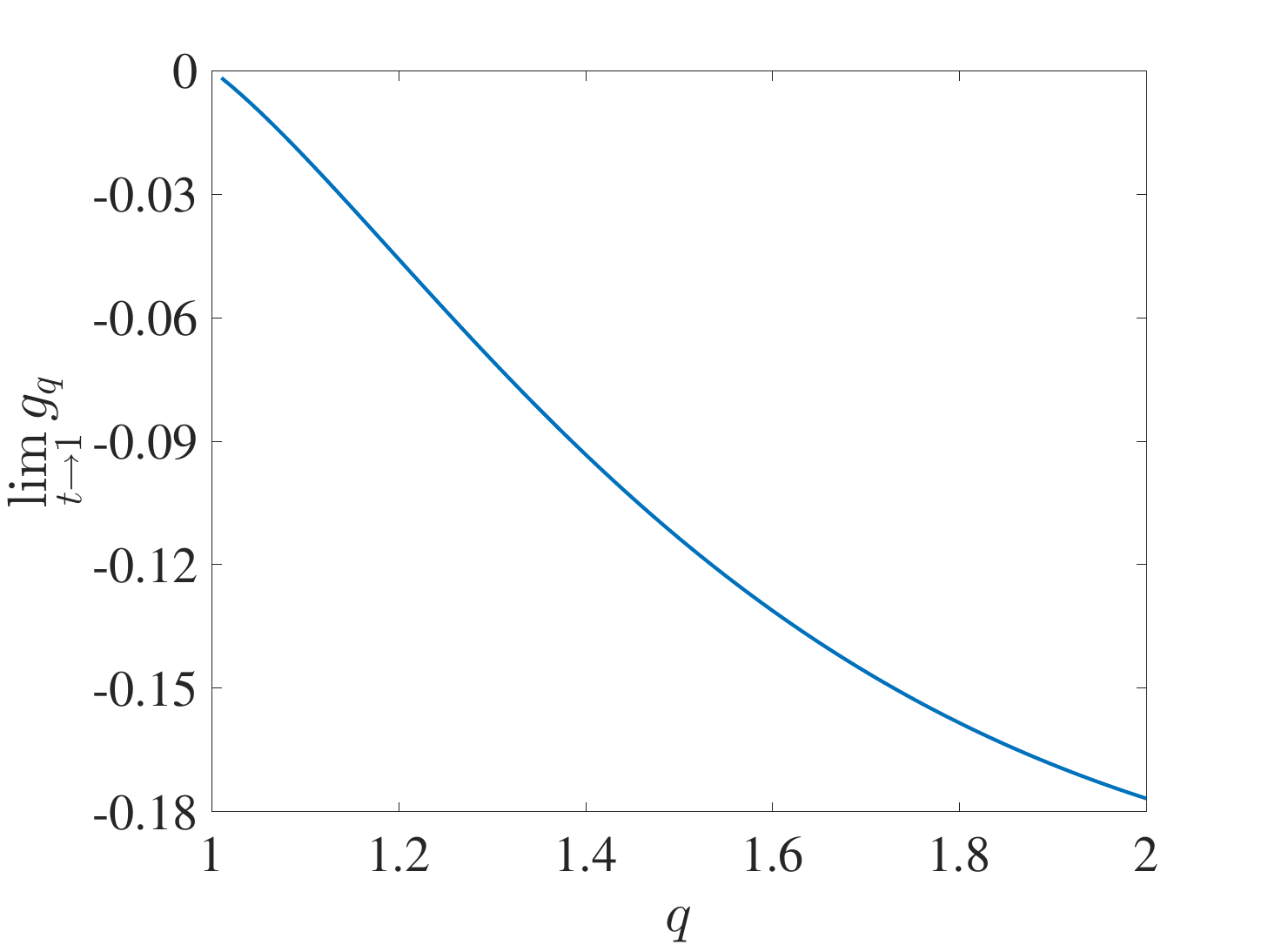}}
\caption{The blue line illustrates that $\lim\limits_{t\rightarrow1}g_q=\lim\limits_{t\rightarrow1}\frac{\partial^2h_q}{\partial t^2}$ is non-positive for $1<q\leq2$.\\}\label{fig 2}
\end{figure}
\begin{equation*}
\begin{aligned}
\nabla M(t,q)=\bigg(\frac{\partial M(t,q)}{\partial t},\frac{\partial M(t,q)}{\partial q}\bigg)=0\\
\end{aligned}
\end{equation*}
in the domain $D^\circ=\{(t,q)|0<t<1,1<q<2\}$, as shown in Fig. \ref{fig 1}, which suggests that there are no critical points within $D^\circ$. Thus, the minimal and maximal values only occur at the boundary points, which are
\begin{equation*}
\begin{aligned}
&\lim\limits_{t\rightarrow0}g_q=-\infty,\\
&\lim\limits_{t\rightarrow1}g_q=-\frac{(q-1)(2^q\cdot q^2-5q\cdot 2^q+6\cdot 2^q+4q^2-2q-6)}{12(2^q-2)^{\frac{2q-1}{q}}},\\
\end{aligned}
\end{equation*}
and the Fig. \ref{fig 2} shows $\lim\limits_{t\rightarrow1}g_q\leq0$ for $1<q\leq 2$. This means that the function $h_q(C^2)$ is concave with respect to $C^2$ for $1<q\leq2$.

\section{Proof of Lemma 2}\label{B}
First, we prove that the function $h_q^2(C^2)$ is monotonically increasing with respect to $C^2$. Set $t=C^2$,
by calculation, we get its first-order partial derivative
\begin{equation*}
\begin{aligned}
\frac{\partial h_q^2}{\partial t}=&\frac{1}{2^{q}}\bigg[1-\bigg(\frac{1+\sqrt{1-t}}{2}\bigg)^q-\bigg(\frac{1-\sqrt{1-t}}{2}\bigg)^q\bigg]^{\frac{2}{q}-1}\\
&\times\frac{(1+\sqrt{1-t})^{q-1}-(1-\sqrt{1-t})^{q-1}}{\sqrt{1-t}}.\\
\end{aligned}
\end{equation*}
Obviously, $\frac{\partial{h}_q^2}{\partial t}$ is greater than 0 for $0<t<1$ and $q>1$, which indicates that ${h}_q^2(C^2)$ is a monotonically increasing function with regard to $C^2$ for $0\leq C^2\leq1$.

Next, we verify that $h_q^2(C^2)$ is a convex function of $C^2$, that is, prove that
\begin{equation*}
\begin{aligned}
\widetilde{g}_q=\frac{\partial^2{h}_q^2}{\partial t^2}\geq0.
\end{aligned}
\end{equation*}
Through derivation, one has
\begin{equation*}
\begin{aligned}
\widetilde{g}_q=&\frac{1}{2^{q}}\bigg[1-\bigg(\frac{1+\sqrt{1-t}}{2}\bigg)^q-\bigg(\frac{1-\sqrt{1-t}}{2}\bigg)^q\bigg]^{\frac{2}{q}-2}\\
&\times M'(t,q),\\
\end{aligned}
\end{equation*}
where $M'(t,q)=\frac{2-q}{2^{q+1}}\xi_1+\xi_2(\xi_3-\xi_4)$, $\xi_1,\xi_2,\xi_3,\xi_4$ are respectively given in Eq. (\ref{le1}).

\begin{figure}[htbp]
\centering
{\includegraphics[width=8cm,height=6cm]{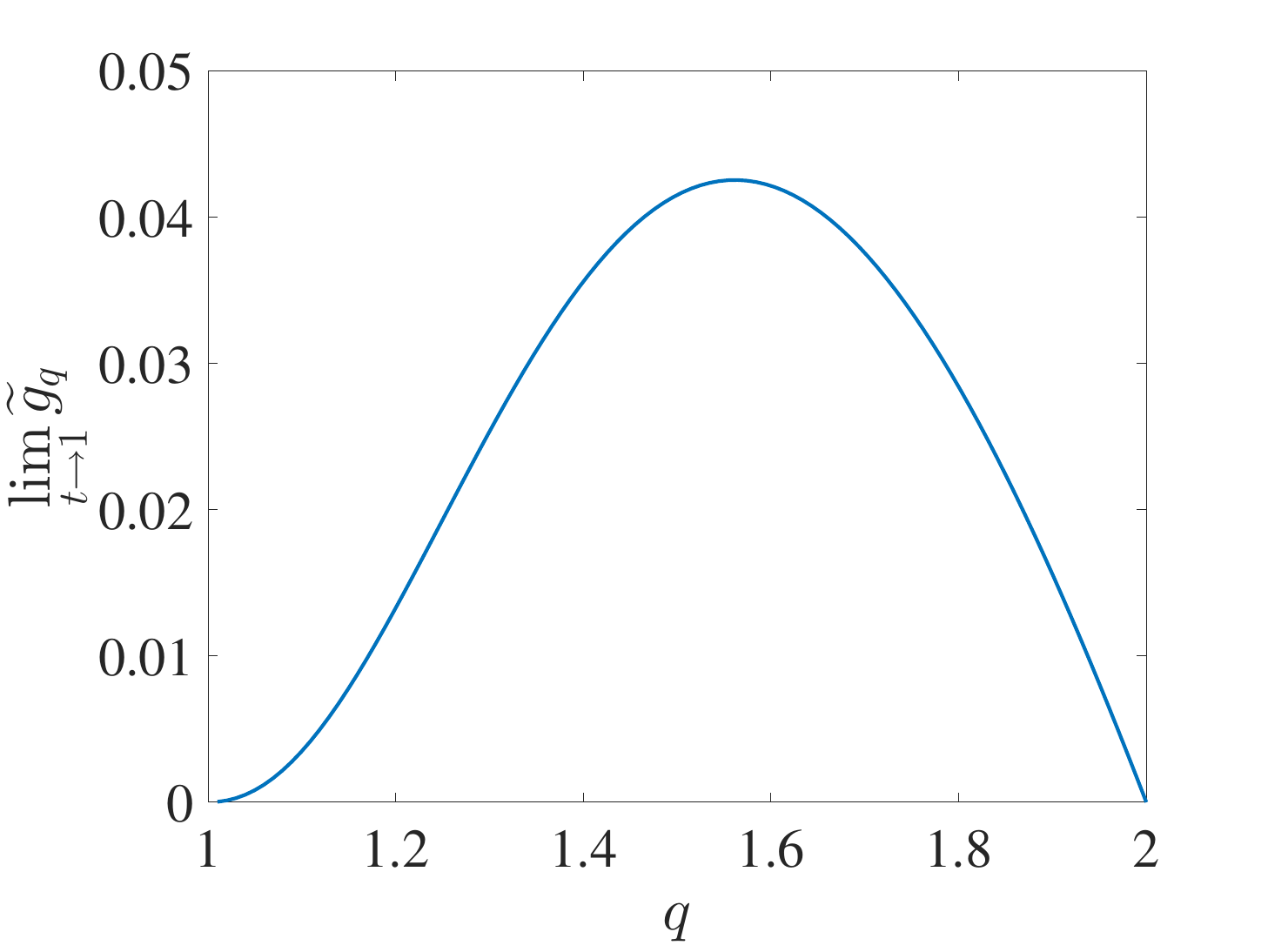}}
\caption{The blue line states that $\lim\limits_{t\rightarrow1}\widetilde{g}_q=\lim\limits_{t\rightarrow1}\frac{\partial^2{h}_q^2}{\partial t^2}\geq0$ for $1<q\leq2$.\\}\label{fig 4}
\end{figure}
Following similar approaches to the proof of Appendix \ref{A}, we first analyze the points within the region $D^\circ$. It is not difficult to obtain that the gradient $\nabla M'(t,q)=0$ is not solvable within the domain $D^\circ$. Therefore, the maximum or minimum must be taken at the boundary points. Meanwhile, we get that the limits of two boundary points are
\begin{equation*}
\begin{aligned}
&\lim\limits_{t\rightarrow0}\widetilde{g}_q=+\infty,\\
&\lim\limits_{t\rightarrow1}\widetilde{g}_q=-\frac{(q^2-3q+2)(4q+q\cdot2^q-3\cdot2^q)}{12(2^q-2)^\frac{2q-2}{q}},\\
\end{aligned}
\end{equation*}
where $\lim\limits_{t\rightarrow1}\widetilde{g}_q\geq0$, as shown in Fig. \ref{fig 4}. This suggests $\widetilde{g}_q\geq0$ in the region $D$, namely, the function $h_q^2(C^2)$ is convex with respect to $C^2$.

\end{appendix}


\end{document}